\documentclass[submission]{eptcs}
\providecommand{\event}{TARK 2023}

\usepackage{breakurl}             
\usepackage{underscore}           

\usepackage{amsfonts}
\usepackage{mathrsfs}
\usepackage{amssymb}
\usepackage{amsthm}
\usepackage{latexsym}
\usepackage{graphicx}
\usepackage{enumerate}
\usepackage{mathtools}
\usepackage{color}
\usepackage{float}
\usepackage{array}
\usepackage{eucal}

\newcommand{\shortv}[1]{}

\usepackage{tabularx}

\mathtoolsset{centercolon}

\usepackage{amssymb,tikz}

\newcommand{\defin}[1]{\textbf{#1}}

\newcommand{\lthen}{\rightarrow}
\newcommand{\liff}{\leftrightarrow}

\newcommand{\dimp}{\Leftrightarrow}

\newcommand{\citeyear}{\cite}

\newcommand{\sel}{\textit{sel}}

\newcommand{\val}[1]{[\![ #1 ]\!]}
\newcommand{\commentout}[1]{}

\newcommand{\mult}[1]{\{\!\!\{ #1 \}\!\!\}}

\renewcommand{\phi}{\varphi}

\newcommand{\dop}{\mathit{do}}

\newcommand{\noop}{\mathit{no\mbox{-}op}}
\newcommand{\id}{\mathit{id}}

\renewcommand{\L}{\mathcal{L}}
\newcommand{\A}{\mathcal{A}}

\newtheorem{theorem}{Theorem}

\title{Sequential Language-based Decisions}
\author{Adam Bjorndahl
\institute{Department of Philosophy\\
Carnegie Mellon University\\
Pittsburgh, USA}
\email{abjorn@cmu.edu}
\and
Joseph Y. Halpern
\institute{Department of Computer Science\\
Cornell Univeristy\\
Ithaca, USA}
\email{halpern@cs.cornell.edu}
}
\def\titlerunning{Sequential Language-based Decisions}
\def\authorrunning{A.~Bjorndahl \& J.Y.~Halpern}

\begin{document}

\maketitle

\begin{abstract}
In earlier work, we introduced the framework of \emph{language-based
decisions}, the core idea of which was to modify Savage's classical
decision-theoretic framework \cite{Savage} by taking actions to be  
descriptions in some language, rather than
functions from states to outcomes, as they are defined
classically. Actions had
the form ``if $\psi$ then 
$\dop(\phi)$'', where $\psi$ and $\phi$ were formulas in some
underlying language, 
specifying what effects would be brought about under what circumstances.
The earlier work allowed only one-step actions. But, in practice,
plans are typically composed of a sequence of steps.  Here, we extend
the earlier framework to \emph{sequential} actions, making it much
more broadly applicable. 
Our technical contribution is a representation theorem in
the classical spirit: agents whose preferences over actions satisfy
certain constraints can be modeled as if they are expected utility
maximizers. As in the earlier work, due to the language-based
specification of the actions, the representation theorem requires a
construction not only of the probability and utility functions
representing the agent's beliefs and preferences, but also the state
and outcomes spaces over which these are defined, as
well as a ``selection function'' which intuitively captures how agents
disambiguate coarse descriptions. The (unbounded) depth of action
sequencing adds substantial interest (and complexity!) to the proof. 
\end{abstract}

\section{Background and motivation}

In earlier work, we introduced the framework of \emph{language-based decisions}
\cite{BH21}, the core idea of which was to modify Savage's classical
decision-theoretic framework \cite{Savage} by taking actions to be  
descriptions in some language, rather than
functions from states to outcomes, as they are defined
classically. Actions had the form ``if $\phi$ then 
$\dop(\psi)$'', where $\phi$ and $\psi$ were formulas in some
underlying language, specifying what effects would be brought about
under what circumstances.%
\footnote{This work in turn extended previous work by Blume,
Easley, and Halpern \cite{BEH06} in which the tests in actions, but not the
effects of actions, were specified in a formal language.}
For example, a statement like ``If there is a budget surplus then
$\dop(MW = 15)$ else $\noop$'' would be an action in this framework,
where $MW = 15$ represents  the minimum wage being \$15, and $\noop$
is the action of doing nothing.
The effect of the action $\dop(MW=15)$ is to bring about a state where
the minimum wage is \$15.  But this does not completely specify the
state. (Do businesses close? Is there more automation so jobs are
lost? Are no jobs lost and more people move into the state?)

In this context, we proved a representation theorem in the classical
spirit: agents whose preferences over actions satisfy certain
constraints can be modeled as if they are expected utility
maximizers. This requires constructing not only probability and
utility functions (as is done classically), but also the state and
outcome spaces on which these functions are defined, and a
\emph{selection function} that describes which state will result from
an underspecified action like $\dop(MW=15)$. In this
construction the state and
outcome spaces coincide; intuitively, this
is because the tests that determine whether an action is
performed (``If $\phi$ then...'') and the actions themselves
(``$\dop(\psi)$'') are described using the same language.

The earlier work allowed only one-step actions. But, in practice,
plans are typically composed of a sequence of steps, and we must
choose among such plans: Do I prefer to walk to the cafe and then
call my friend if the cafe is open, or would it be better to call my
friend first, then walk to the cafe and call them back if it's closed?
Should I ring the doorbell once, or ring it once and then a second
time if no one replies to the first? 
Here, we extend the earlier framework to \emph{sequential} actions,
making it much more broadly applicable.

At a technical level, a decision-theoretic framework in which the state and outcome spaces coincide is the perfect setting in which to implement sequential actions, since---given that the actions are understood as functions---we have an immediate and natural way to ``put them in sequence'', namely, by composing the corresponding functions.


Our contribution in this paper is, first, to lay the mathematical
groundwork for reasoning about sequential, language-based actions
(Section \ref{sec:slba}), and second, to prove a representation
theorem analogous to earlier such results (Section \ref{sec:rep}):
roughly speaking, agents whose preferences over sequential actions
satisfy certain axioms can be understood as if their preferences are
derived by maximizing the expected value of a suitable utility
function. Proving this result is substantially
harder in the present setting, owing to the more complex nature of
sequential actions (including but not limited to the fact that we
allow sequential nesting to be arbitrarily deep). The reader is thus
forewarned that the main result depends on a fairly lengthy,
multi-stage proof.

\section{Sequential language-based actions} \label{sec:slba}

The framework presented in this section is an expansion of that
developed in \cite{BH21}. We begin with the same simple, formal
language: let $\Phi$ denote a finite set of \emph{primitive
propositions}, and $\L$ the propositional language consisting of all
Boolean combinations of these primitives. A \defin{basic model (over
  $\L$)} is a tuple $M = (\Omega, \val{\cdot}_{M})$ where $\Omega$ is
a nonempty set of \emph{states} and $\val{\cdot}_{M}: \Phi \to
2^{\Omega}$ is a \emph{valuation function}. This valuation is
recursively extended to all formulas in $\L$ in the usual way, so that
intuitively, each formula $\phi$ is ``interpreted'' as the ``event''
$\val{\phi}_{M} \subseteq \Omega$. We sometimes drop the subscript
when the model is clear from context, and write $\omega \models \phi$
rather than $\omega \in \val{\phi}$. We say that $\phi$ is \emph{satisfiable
in $M$} if $\val{\phi}_{M} \neq \emptyset$ and that $\phi$ is
\emph{valid in $M$} if $\val{\phi}_{M} = \Omega$; we write $\models
\phi$ to indicate that $\phi$ is valid in all basic models. 

Given a finite set of formulas $F \subseteq \L$, the set of \defin{(sequential) actions (over $F$)}, denoted by $\A_{F}$, is defined recursively as follows:
\begin{enumerate}[(1)]
\item
for each $\phi \in F$, $\dop(\phi)$ is an action (called a \emph{primitive action});
\item
$\noop$ is an action (this is short for ``no operation''; intuitively, it is a ``do nothing'' action);
\item
for all $\psi \in \L$ and $\alpha, \beta \in \A_{F}$, not both
$\noop$, \textbf{if $\psi$ then $\alpha$ else $\beta$} is an action;
\item
  for all $\alpha, \beta \in \A_{F}$, not both $\noop$,
  $\alpha;\beta$ is an action (intuitively, this is the action ``do $\alpha$ and
  then do $\beta$''). 
\end{enumerate}
In \cite{BH21}, actions were defined only by clauses (1) and (3).
The idea of ``sequencing'' actions is of course not new;
the semicolon notation is standard in programming languages.

It will also be useful for our main result to have a notion of the \emph{depth} of an action, which intuitively should capture how deeply nested the sequencing is. We do so by induction. The only \defin{depth-$0$} action is $\noop$.
A \defin{depth-$1$ action} is either (1) $\noop$; (2) a primitive
action $\dop(\phi)$; or (3) an action of the form \textbf{if $\psi$ then
$\alpha$ else $\beta$}, where $\alpha$ and $\beta$ are depth-$1$ 
actions.
Now suppose that we have defined depth-$k$ actions for $k \geq 1$; a
\defin{depth-($k+1$)}
action is either (1) a depth-$k$ action; (2) an action of the form
\textbf{if $\psi$ then $\alpha$ else $\beta$}, where $\alpha$ and
$\beta$ are depth-($k+1$) actions; or (3) an action of the form
$\alpha;\beta$, where $\alpha$ is a depth-$k_1$ action, $\beta$ is a
depth-$k_2$ action, and $k_1 + k_2 \le k+1$. Note that we have
defined depth in such a way that the depth-$k$ actions
include all the depth-$k'$ actions for $k' < k$,
and so that \textbf{if...then} constructions do not increase depth---only sequencing does.

As in \cite{BH21}, given a basic model $M = (\Omega,
\val{\cdot}_{M})$, we want $\dop(\phi)$ to correspond to a function
from $\Omega$ to $\Omega$ whose range is contained in
$\val{\phi}_{M}$. For this reason we restrict our attention to basic
models in which each $\phi \in F$ is satisfiable, so that
$\val{\phi}_{M} \neq \emptyset$; call such models
\defin{$F$-rich}. Moreover, in order for $\dop(\phi)$ to pick out a
\textit{function}, we need some additional structure that determines,
for each $\omega \in \Omega$, which state in $\val{\phi}_{M}$ the
function corresponding to $\dop(\phi)$ should actually map to. This is
accomplished using a \emph{selection function} $\sel: \Omega \times
F \to \Omega$ satisfying $\sel(\omega,\phi) \in \val{\phi}_M$. 

The intuition for selection functions is discussed in greater detail
in \cite{BH21}. Briefly: $\dop(\phi)$ says that $\phi$ should be made
true, but there may be many ways of making $\phi$ true (i.e., many
states one could transition to in which $\phi$ is true); $\sel$ picks out
which of these $\phi$-states to actually move to. In this way we can
think of $\sel$ as serving to ``disambiguate'' the meaning of the
primitive actions, which are inherently underspecified. 

Note that selection functions are formally identical to the mechanism
introduced by Stalnaker \citeyear{Stalnaker68} to interpret
counterfactual conditionals. In our context, we can think of a
selection function as another component of an agent's model of the
world, to be constructed in the representation theorem: in addition to
a probability measure (to represent their beliefs) and a utility
function (to capture their preferences), we will also need a selection
function (to specify how they interpret actions). 

A \defin{selection model (over $F$)} is an $F$-rich basic model $M$
together with a selection function $\sel: \Omega \times F \to
\Omega$ satisfying $\sel(\omega,\phi) \in \val{\phi}_M$. 
Given a selection model $(M,\sel)$
over $F$, we define the \emph{interpretation} of $\dop(\phi)$
to be the function $\val{\dop(\phi)}_{M,\sel}: \Omega \to \Omega$
given by:  
$$\val{\dop(\phi)}_{M,\sel}(\omega) = \sel(\omega, \phi).$$
This interpretation can then be extended to all sequential actions in $\A_{F}$ in the obvious way: 
$$
\val{\textbf{if $\psi$ then $\alpha$ else $\beta$}}_{M,\sel}(\omega) = \begin{cases}
\val{\alpha}_{M,\sel}(\omega) & \textrm{if $\omega \in \val{\psi}$}\\
\val{\beta}_{M,\sel}(\omega) & \textrm{if $\omega \notin \val{\psi}$,}
\end{cases}
$$
and
$$\val{\alpha;\beta}_{M,\sel} = \val{\beta}_{M,\sel} \circ \val{\alpha}_{M,\sel}.$$

\section{Representation} \label{sec:rep}

Let $\succeq$ be a binary relation on $\A_{F}$, where we understand
$\alpha \succeq \beta$ as saying that $\alpha$ is ``at least as good
as'' $\beta$ from the agent's subjective perspective. Intuitively,
such a binary relation is meant to be reasonably ``accessible'' to
observers, ``revealed'' by how an agent chooses between binary
options. As usual, we define $\alpha \succ \beta$ as an abbreviation
of $\alpha \succeq \beta$ and $\beta \not\succeq \alpha$, and $\alpha
\sim \beta$ as an abbreviation of $\alpha \succeq \beta$ and $\beta
\succeq \alpha$; 
intuitively, these relations represent ``strict preference'' and
``indifference'', respectively. 

We assume that $\succeq$ is a \emph{preference order}, so is \emph{complete}
(i.e., for all acts
$\alpha, \beta \in \A_{F}$, either $\alpha \succeq \beta$ or $\beta
\succeq \alpha$) and transitive.
Note that completeness immediately gives reflexivity as well.
While there are good philosophical reasons to
consider incomplete relations (see \cite{DMO04} and the references therein),
for the purposes of this paper we adopt the assumption of completeness in order to simplify the (already quite involved) representation result.

A \defin{language-based SEU (Subjective Expected Utility)
representation} for a relation $\succeq$ on $\A_{F}$ is a
selection model $(M,\sel)$ together with a probability measure
$\Pr$ on 
$\Omega$ and a utility function $u: \Omega \to \mathbb{R}$ such
that, for all $\alpha, \beta \in \A_{F}$,  
\begin{equation} \label{eqn:rep}
\alpha \succeq \beta \dimp \sum_{\omega \in \Omega} \Pr(\omega) \cdot u(\val{\alpha}_{M,\sel}(\omega)) \geq \sum_{\omega \in \Omega} \Pr(\omega) \cdot u(\val{\beta}_{M,\sel}(\omega)).
\end{equation}
Our goal is to show that such a representation exists if the
preference order satisfies one key axiom, discussed below.

\subsection{Canonical maps and canonical actions}

For each $a \subseteq \Phi$, let
$$\phi_{a} = \bigwedge_{p \in a} p \land \bigwedge_{q \notin a} \lnot q,$$
so $\phi_{a}$ settles the truth values of each primitive propositions
in the language $\L$: it says that $p$ is true iff it belongs to
$a$. An \defin{atom} is a formula of the form $\phi_{a}$.%
\footnote{Not to be confused with \emph{atomic propositions}, which is another common name for primitive propositions.}
Since we
are working with a classical propositional logic, it follows that for
all formulas $\phi \in \L$ and atoms $\phi_{a}$, the truth of $\phi$
is determined by $\phi_{a}$: either $\models \phi_{a} \lthen \phi$, or
$\models \phi_{a} \lthen \lnot \phi$. In the framework of \cite{BH21},
it followed that 
every action could be identified with a function from atoms to
elements of $F$, since atoms determine whether the tests in an action
hold. In our context, however, things are not so simple: actions
can be put in sequence, so even though an atom may tell us which
tests at the ``first layer'' hold, so to speak, it may not
be enough to tell us which later tests hold. For example, in an
action like ``if $p$ then $\dop(r)$ else $\dop(r')$; if $q$ then
  $\dop(\neg r)$'', the atom that currently holds determines whether
  $p$ holds, but tells us nothing about whether $q$ will hold when we
  get around to doing the second action in the sequence.

To deal with this, we need an outcome space that is richer than just $F$
  (i.e., richer than the set of all primitive actions); roughly
speaking, we will instead identify actions with functions from atoms
to ``canonical'' ways of describing the sequential structure of
actions. We now make this precise. 

Suppose that $|2^{\Phi}| = N$, so there are $N$ atoms; call them
$a_{1}, \ldots, a_{N}$. For each subset $A$ of atoms, let
$\phi_A = \bigvee_{a \in A} \phi_{a}$. A
basic fact of propositional logic is that for every formula $\phi$,
there is a unique set $A$ of atoms such that $\phi$ is logically
equivalent to $\phi_A$. Let $\tilde{F} = \{\phi_{A} \: : \: (\exists
\phi \in F)(\models \phi_{A} \liff \phi)\}$.

We want to associate with each action $\alpha$ of depth $k$ a
\defin{canonical action} $\gamma_\alpha$ of depth $k$ that is,
intuitively, equivalent to $\alpha$.  The canonical action
$\gamma_\alpha$ makes explicit how $\alpha$ acts in a state
characterized by an atom $a$. We define $\gamma_\alpha$ by induction
on the structure of $\alpha$. It is useful in the construction to
simultaneously define the \defin{canonical map}
$c_\alpha$ associated with $\alpha$, a function  from atoms to actions
such that, for all atoms $a$, $c_\alpha(a)$ has the form $\noop$,
$\dop(\phi_A)$, or 
$\dop(\phi_A); \gamma_\beta$ for some set $A$ of atoms and action $\beta$.
Intuitively, $c_\alpha$ 
defines how $\alpha$ acts in a state characterized by an atom $a$.  For
example, if $\alpha$ is
\textbf{if $a$ then $\dop(\phi_A)$ else $\beta$}, then
$c_\alpha(a) = \dop(\phi_A)$.

If $\alpha = \noop$, then $\gamma_\noop = \noop$ and $c_\noop$
is the constant function such that $c_\noop(a) = \noop$ for all atoms
$a$.
If $\alpha$ is a depth-1 action other than $\noop$, then we define
$c_\alpha$ by induction on structure:
\begin{eqnarray*}
c_{\dop(\phi)}(a) & = & \dop(\phi_{A}), \textrm{ where $A$ is the unique subset of atoms such that $\models \phi_{A} \liff \phi$}\\
c_{\textbf{if $\psi$ then $\alpha$ else $\beta$}}(a) & = & \begin{cases}
c_{\alpha}(a) & \textrm{if $\models \phi_{a} \lthen \psi$}\\
c_{\beta}(a) & \textrm{if $\models \phi_{a} \lthen \lnot \psi$}.
\end{cases}\\
\end{eqnarray*}
The action $\gamma_\alpha$ is the depth-1 action defined as follows:
$$\gamma_{\alpha} = \textbf{if $\phi_{a_{1}}$ then $c_\alpha(a_{1})$ else (if
  $\phi_{a_{2}}$ then $c_\alpha(a_{2})$ else ($\cdots$(if $\phi_{a_{N-1}}$
  then $c_\alpha(a_{N-1})$ else $c_\alpha(a_{N})$)$\cdots$))}$$ 
if at least one of $c_\alpha(a_{N-1})$ or $c_\alpha(a_N)$ is not $\noop$.
If both are $\noop$, then  \textbf{ if $\phi_{a_{N-1}}$
    then $c_\alpha(a_{N-1})$ else $c_\alpha(a_{N})$} is not an action
according to our definitions; in this case, we take 
$$\gamma_{\alpha} = \textbf{if $\phi_{a_{1}}$ then $c_\alpha(a_{1})$ else (if
  $\phi_{a_{2}}$ then $c_\alpha(a_{2})$ else ($\cdots$(if $\phi_{a_{m}}$
  then $c_\alpha(a_{m})$ else $\noop$))},$$
where $m$ is the least index such that $c_\alpha(a_{m}) \ne \noop$.  (If
$c_\alpha(a_m) = \noop$ for all $m$, then $\gamma_\alpha = \noop$.)

If $\alpha$ is a depth-$(k+1)$ action other than $\noop$, then we
  again define $c_\alpha$ by induction on structure:
  \begin{eqnarray*}
c_{\textbf{if $\psi$ then $\alpha$ else $\beta$}}(a) & = & \begin{cases}
c_{\alpha}(a) & \textrm{if $\models \phi_{a} \lthen \psi$}\\
c_{\beta}(a) & \textrm{if $\models \phi_{a} \lthen \lnot \psi$}
\end{cases}\\
c_{\alpha;\beta}(a) & = & \begin{cases}
c_{\beta}(a) & \textrm{if $c_{\alpha}(a) = \noop$}\\
\dop(\phi_{A});\gamma_{\beta} & \textrm{if $c_{\alpha}(a) = \dop(\phi_{A})$}\\
\dop(\phi_{A});\gamma_{\beta';\beta} & \textrm{if
  $c_{\alpha}(a) = \dop(\phi_{A});\gamma_{\beta'}$.} 
\end{cases}
\end{eqnarray*}
The canonical action $\gamma_\alpha$ is defined as above for the
depth-1 case.

We take $CA^k$ to be the set of canonical actions of depth $k$, and 
$CM^k$ to be the set of canonical maps that correspond to
some depth-$k$ action.
Finally, let $CA^{k,-}$ consist of all depth $k$-actions of the form
$\noop$, $\dop(\phi_A)$, or $\dop(\phi_A);\gamma_\beta$,
where $\beta$ is a depth ($k-1$)-action.
Note that if $\alpha$ is a depth-$k$ action and $a$ is an atom, then
$c_\alpha(a) \in CA^{k,-}$.
Observe that since the set of atoms is finite, as is $\tilde{F}$, it
follows that for all $k$, $CM^k$, $CA^k$, and $CA^{k,-}$ are also finite.
This will be crucial in our representation proof.

\subsection{Cancellation}

As in \cite{BH21,BEH06}, the key axiom in our representation theorem
is what is known as a \emph{cancellation axiom},
although the details differ due to the nature of our actions. 
Simple versions of the cancellation axiom go back to
\cite{KLST71,scott64};
our
version, like those used in
\cite{BH21,BEH06}, has more structure. See \cite{BEH06} for further
discussion of the axiom.

The axiom uses multisets.
Recall that a \emph{multiset},
intuitively, is a set that allows for multiple instances of each of
its elements. Thus two multisets are equal just in case they contain
the same elements \textit{with the same multiplicities}. We use
``double curly brackets'' to denote multisets, so for instance
$\mult{a,b,b}$ is a multiset, and it is distinct from $\mult{a,a,b}$:
both have three elements, but the mulitiplicity of $a$ and $b$ differ.
With that background, we can state the axiom:  
\begin{description}
\item[(Canc)]
Let $\alpha_{1}, \ldots, \alpha_{n}, \beta_{1}, \ldots, \beta_{n} \in \A_{F}$, and suppose that for each $a \subseteq \Phi$ we have
$$\mult{c_{\alpha_{1}}(a), \ldots, c_{\alpha_{n}}(a)} = \mult{c_{\beta_{1}}(a), \ldots, c_{\beta_{n}}(a)}.$$
Then, if for all $i < n$ we have $\alpha_{i} \succeq \beta_{i}$, it follows that $\beta_{n} \succeq \alpha_{n}$.
\end{description}
Intuitively, this says that if we get the same outcomes (counting multiplicity) using the canonical maps for $\alpha_1, \ldots, \alpha_n$ as for $\beta_1, \ldots, \beta_n$ in each state, then we should view the collections $\mult{\alpha_1, \ldots, \alpha_n}$ and $\mult{\beta_1, \ldots, \beta_n}$ as being ``equally good'', so if $\alpha_i$ is at least as good as $\beta_i$ for $i=1, \ldots, n-1$, then, to balance things out, $\beta_n$ should be at least as good as $\alpha_n$. How intuitive this is depends on how intuitive one finds the association $\alpha \mapsto c_{\alpha}$ defined above; if the map $c_{\alpha}$ really does capture ``everything decision-theoretically relevant'' about the action $\alpha$, then cancellation does seem reasonable.

In particular, it is not hard to show that whenever $\alpha$ and
$\beta$ are such that $\gamma_{\alpha} = \gamma_{\beta}$
(which of course is equivalent to $c_{\alpha} = c_{\beta}$),
cancellation implies
that $\alpha \sim \beta$. In other words, any information about
$\alpha$ and $\beta$ that is lost in the transformation to canonical
actions is also forced to be irrelevant to decisionmaking. This means
that {\bf (Canc)} entails, among other things, that agents do not
distinguish between logically equivalent formulas (since, e.g., when
$\models \phi \liff \phi'$, it's easy to see that
$\gamma_{\dop(\phi)} = \gamma_{\dop(\phi')}$). 


\subsection{Construction}

\begin{theorem}
If $\succeq$ is a
preference order
on $\A_{F}$ satisfying {\bf (Canc)},
then there is a language-based
subjective expected utility representation of $\succeq$.
\end{theorem}

\begin{proof}
As in \cite{BH21}, we begin by following the proof in \cite[Theorem
  2]{BEH06}, 
which says that if a preference order on a set of acts mapping a finite
state space to a finite outcome space satsifies the
cancellation axiom, then it has a state-dependent representation.
``State-dependent'' here means that the utility
function constructed depends jointly on both states and outcomes, in a
sense made precise below.
To apply this theorem in our setting, we first fix $k$ and take $CM^k$
to be the set of acts. With this viewpoint, the state space is the
set of atoms and the outcome space is $CA^{k,-}$; as we observed, both are
finite. 

The relation $\succeq$ on $\A_{F}$ induces a relation $\succeq^{k}$ on
$CM^{k}$ defined in the natural way:  
$$c_{\alpha} \succeq^{k} c_{\beta} \dimp \alpha \succeq \beta.$$
As noted, {\bf (Canc)} implies that $\alpha \sim \alpha'$ whenever
$c_{\alpha} = c_{\alpha'}$, from which it follows that $\succeq^{k}$
is well-defined.
To apply Theorem 2 in \cite{BEH06}, it must also be the case that
$\succeq^k$ is a 
preference order and satisfies cancellation, which is
immediate from the definition of $\succeq^k$ and the fact that
$\succeq$ is a preference order and satisfies cancellation.
It therefore follows that $\succeq^k$ has a state-dependent
representation; that is, there exists a real-valued utility
function $v^k$ defined on state-outcome pairs such that, for all
depth-$k$ actions $\alpha$ and $\beta$,
\begin{equation}\label{eq1}
  \mbox{$c_\alpha \succeq^k c_\beta$ iff $\sum_{i=1}^N v^k(a_i,c_{\alpha}(a_i))
  \ge \sum_{i=1}^N v^k(a_i,c_{\beta}(a_i))$}.a
\end{equation}
It follows from our definitions that for all
depth-$k$ actions $\alpha$ and $\beta$,
$$\mbox{$\alpha \succeq \beta$ iff $\sum_{i=1}^N v^k(a_i,c_{\alpha}(a_i))
  \ge \sum_{i=1}^N v^k(a_i,c_{\beta}(a_i))$}.$$
As we observed, we needed to restrict to depth-$k$ actions here in order to
ensure that the outcome space is finite, which is necessary to apply
Theorem 2 in \cite{BEH06}.
%

Our next goal is to define a selection model $M = (\Omega^k,
\val{\cdot}_{M}, \sel)$, a probability $\Pr^k$ on $\Omega^k$, and a utility function $u^k$ 
on $\Omega^k$ such that, for all actions $\alpha$ and $\beta$ of depth $k$,
\begin{equation}\label{eq1.5}
\alpha \succeq \beta \mbox{ iff }
\sum_{\omega \in \Omega^k} {\Pr}^k(\omega)u^k(\val{\alpha}_{M,\sel}(\omega)) \ge
\sum_{\omega \in \Omega^k}{\Pr}^k(\omega)u^k(\val{\beta}_{M,\sec}(\omega)).
\end{equation}
Eventually, we will construct a single (state and outcome) space $\Omega^*$,
a probability $\Pr^*$ on $\Omega^*$, and a utility $u^*$ on $\Omega^*$
that we will use to provide a single representation theorem for all
actions, without the restriction to depth $k$, but we seem to need to
construct the separate spaces first.

As a first step to defining $\Omega^k$, define
a \emph{labeled $k$-tree} to be a balanced tree of depth $k$ whose
root is labeled by an atom such that each non-leaf node
has exactly $N$ children, labeled $a_1, \ldots,
a_N$, respectively.
An \emph{ordered labeled $k$-tree ($k$-olt)} is a labeled
$k$-tree where, associated with each non-leaf node, there is a total 
order on its children.
We assume that in different labeled $k$-trees, the
nodes come from the same set, and corresponding 
nodes have the same label, so there is a unique labeled
$k$-tree and $k$-olts differ only in the total order associated with
each non-leaf node and the label of the root.
Let $T^k$ consist of all $k$-olts.
For $k' \ge k$, a $(k')$-olt $s^{k'}$ \emph{extends} (or
\emph{is an extension of}) a $k$-olt $s^k$ if $s^k$ is the prefix
of $s^{k'}$ of depth $k$; we call $s^k$ the \emph{projection} of
$s^{k'}$ onto depth $k$.

The intuition behind a $k$-olt is the following: the atom
associated with the root $r$ describes what is true before an
action is taken.
For each non-leaf node $t$, the total order associated with $t$
on the children of $t$ describes the selection function at $t$
(with children lower in the order considered ``closer''). 
For example, suppose
that there are two primitive propositions, $p$ and $q$.  Then there
are four atoms.  If we take the action $\dop(\phi)$ starting at $r$, we want to
``move'' to the ``closest'' child of $r$ satisfying $\phi$,
which is the child lowest in the ordering associated with $r$.
For example, suppose
that the total order on the atoms associated with $r$ is
$\neg p \land q < \neg p \land \neg q < p \land \neg q < p \land q$.
Then if we take the action $\dop(p \lor q)$
starting at $r$, we
move to the child labeled with the atom $\neg p \land q$; if we
instead do $\dop 
(p \lor \neg q)$, we move to the child labeled $\neg p \land
\neg q$; and if instead we do \textbf{if $q$ then $\dop(p \lor q)$
  else $\dop(p \lor \neg q)$}, which of these two children we move to
depends on whether $q$ is true at the atom labeling $r$.  For an
action $\dop(p \lor q); \dop(p \lor \neg q)$, we move further down
the tree.  The first action,
$\dop(p \lor q)$, takes us to the child $t$ of $r$ labeled $\neg
p \land q$.  We then take the action $\dop(p \lor \neg q)$ from there,
which gets us to a child of $t$.  Which one we get to depends on the
ordering of the children of $t$
associated with $t$.

It turns out that our states must be even richer than this; they must
in addition include a \emph{$k$-progress function} $g$ that maps
each node $t$ in a $k$-olt $s^k$ to a descendant of $t$ in $s^k$.  We
give the intuition behind progress functions shortly. We 
take $\Omega^k$ to consist of all pairs $(s^k,g)$, where $s^k \in T^k$ and $g$
is a $k$-progress function
and for each primitive proposition $p \in \Phi$, we define
$$\val{p} = \{(s^k,g) \: : \: \textrm{$p \in a$, where $a$ labels
  $g(r)$ and $r$ is the root of $s^k$}\}.$$

We now want to associate with each depth-$k$ action $\alpha$ a function
$f_{\alpha}: \Omega^k \to \Omega^k$;
intuitively, this is the transition on states that we want to be induced
by the selection function.
To begin, we define $f_{\alpha}$ only on states of the form $(s^k,\id)$, where $\id$ is the identity function.
We take $f_{\alpha}(s^k,\id) = (s^k,g_{\alpha,s^k})$, where 
$g_{\alpha,s^k}$ is defined formally below.
Intuitively, if $t$ is a node at depth $k'$ of $s^k$, then 
$g_{\alpha,s^k}(t)$ describes the final state if the action
$\alpha$ were to (possibly counterfactually) end up at the node $t$ after running
for $k'$ steps, and then continued running.

Given a $k$-olt $s^k$ whose root $r$ is labeled $a$ and an 
action $\alpha$ of depth at most $k$,
we define $g_{\alpha,s^k}(t)$ by induction on the depth of $\alpha$.
For the base case, we
take $g_{\noop,s^k} = id$.
Now suppose inductively that $\alpha$ has
depth $m$ and we have defined $g_{\alpha',s^k}$ for all actions
$\alpha'$ of depth $m-1$.
There are three cases to consider. (1)
If $c_{\alpha}(a) = \noop$, then
$g_{\alpha,s^k} = id$.
(2)
If $c_{\alpha}(a) = \dop(\phi_A)$, then $g_{\alpha,s^k}(r)$ is the 
``closest'' (i.e., minimal) child $t'$ of $r$ among those labelled
by an atom in $A$,
according to the total order
labeling $r$;
$g_{\alpha,s^k}(t) = t$ for all nodes $t \ne r$.
(3)
If $c_{\alpha}(a) =
\dop(\phi_{A});\gamma_\beta$
(which means $\beta$ is an action of depth at most $m-1$),
then $g_{\alpha,s^k}(r) = g_{\beta,s^{k,t'}}(t')$, where $t' =
g_{\dop(\phi_{A}),s^{k}}(r)$ and
$s^{k,t'}$ is
the $(k-1)$-subolt of $s^k$ rooted at $t'$.
The intuition here is that $g_{\alpha,s^k}(r)$ is supposed to output
the descendent of $r$ that is reached by doing $\alpha$; the fact that
$c_{\alpha}(a) = \dop(\phi_{A});\gamma_\beta$ tells us that the way
$\alpha$ works (in a state where $a$ holds) is by first making
$\phi_{A}$ true, and then following up with $\beta$. This means we
must first move to the ``closest'' child of $r$ where $\phi_{A}$
holds, which is $t'$, and subsequently moving to whichever descendant
of $t'$ that $\beta$ directs us to (which is defined, by the inductive
hypothesis). 
Finally, if $t \ne 
r$, let $t''$ be the first step on the (unique) path from $r$ to $t$ and
let $s^{k,t''}$ be the $(k-1)$-subolt of $s^k$ rooted at
$t''$. Then
$g_{\alpha,s^k}(t) = g_{\beta,s^{k,t''}}(t)$
(where, once again, this is defined by the inductive hypothesis).
This essentially forces us to ``follow'' the unique path from $r$ to
$t$, and then continue from that point by doing whatever the remaining
part of the action $\alpha$ demands. It is clear from this definition that if the root of $s^k$
is labeled by $a$, then $g_{\alpha,s^k} = g_{c_{\alpha}(a),s^k}$. 

We now extend $f_{\alpha}$ to states of the form $(s^k, g_{\beta, s^k})$ by setting $f_{\alpha}(s^k, g_{\beta, s^k}) = f_{\beta;\alpha}(s^{k},id)$. Intuitively, the state $(s^k,g_{\beta, s^k})$ is a
state where $\beta$ has ``already happened'' (i.e., it's the state we would arrive at by doing $\beta$ in $(s^k, id)$) so doing $\alpha$ in this state should be the same as doing first $\beta$ then $\alpha$ in $(s^k, id)$.
 
Observe that a $k$-progress function $g_{\alpha,s^k}$
not only tells us the node that $\alpha$ would reach if it started
at the root of $s^k$,
but also gives a great deal of counterfactual
information
about which nodes would be reached starting from anywhere in $s^k$. 
This is in the same spirit as \emph{subgame-perfect
equilibrium} \cite{Selten75}, which
can depend on what happens at states that are never actually reached
in the course of play, but could have been reached if play had gone
differently. Like this game-theoretic notion, our
representation theorem requires a kind of counterfactual information. 

In light of (\ref{eq1}), to prove (\ref{eq1.5}), it suffices to define
our selection function
$\sel$ so that
$\val{\alpha}_{M,\sel} = f_{\alpha}$, and find $\Pr^k$ and $u^k$ such that for all
actions $\alpha$ of depth $k$,
\begin{equation}\label{eq2}
\sum_{i=1}^N v^k(a_i,c_{\alpha}(a_i)) = \sum_{(s^k,g) \in \Omega^k} {\Pr}^k(s^k,g)
u^k(f_{\alpha}(s^k,g)).
\end{equation}

Our definition of $f_{\alpha}$ is set up to make defining the right
selection function straightforward: we simply set $\sel((s^k,g),\phi)
= f_{\dop(\phi_{A})}(s^k,g)$, where $A$ is the unique set of atoms
such that $\models \phi_{A} \liff \phi$.
It is then easy to check that $\val{\alpha}_{M,\sel} = f_{\alpha}$.

Define $\Pr^k(s^k,g) = 0$ if $g \ne \id$, and 
$\Pr^k(s^k,\id) = 1/|T^k|$ for all $s^k \in T^k$.
Given this, to
establish (\ref{eq2}), it suffices to define $u^{k}$ such that for all
actions $\alpha$ of depth $k$, 
\begin{equation}\label{eq4}
|T^k|\sum_{i=1}^N v^k(a_i,c_{\alpha}(a_i)) = \sum_{s^k \in T^k}
  u^k(f_{\alpha}(s^k,\id)).
  \end{equation}

Given an atom $a$, let $T^k_a$ consist of all $k$-olts whose
root is labeled by $a$.
By definition of $f_{\alpha}$, to prove (\ref{eq4}), it suffices to
prove, for each 
atom $a \in \{a_1, \ldots, a_N\}$ and all actions $\alpha$ of depth $k$, that
\begin{equation}\label{eq3}
|T^k|v^k(a,c_{\alpha}(a)) = \sum_{s^k \in T^k_{a}} u^k(s^k,g_{\alpha,s^k}) = \sum_{s^k \in T^k_{a}} u^k(s^k,g_{c_{\alpha}(a),s^k}),
\end{equation}
where the second equality follows from the fact, observed above, that
$g_{\alpha,s^k} = g_{c_{\alpha}(a),s^k}$ whenever $s^k \in T_a^k$. 


Since $v^k$ is given, for each depth-$k$ action $c_{\alpha}(a)$, the left-hand
side of (\ref{eq3}) is just a number.  
Replace each term $u^k(s^k,g_{c_{\alpha}(a),s^k})$
for $s^k \in T^k_a$
by the variable $x_{s^k,g_{c_{\alpha}(a),s^k}}$.
This gives us a system of linear equations, one for each action
$c_{\alpha}(a)$, with 
variables $x_{s^k,g}$, where the coefficient of $x_{s^k,g}$ in the equation
corresponding to action $\alpha$ is either 1 or 0, depending on whether
$g_{\alpha,s^k} = g$. We want to show that this system has a solution. 

We can describe the relevant equations as the product $MX = U$ of matrices,
where $M$ is a matrix whose entries are either 0 or 1, and $X$ is a
vector of variables (namely, the variables $x_{s^k,g}$).
The matrix $M$ has one row corresponding
to each
action
in $CA^{k,-}$
(since, for all actions 
$\alpha$ of depth $k$, $c_\alpha(a) \in CA^{k,-}$), and one column
corresponding to 
each state $(s^k,g)$ 
with $s^k \in T^k_a$.  
The entry in $M$ in the row corresponding to the action $\gamma_{\alpha}$
and the column corresponding to $(s^k,g)$ 
is 1 if $g_{\gamma_a,s^k} = g$ (i.e., if $f_\alpha(s^k,id) =
(s^k,g_{\alpha,s^k}) = (s^k,g)$) and 0 otherwise.
A basic result of linear algebra tells us that
this system has a
solution if the rows of the matrix $M$ (viewed as vectors) are
independent.  We now show that this is the case. 

Let $r_\alpha$ be the row of $M$ indexed by action $\alpha \in CA^{k,-}$.  
Suppose that a linear combination of rows is 0; that is,
$\sum_{\alpha} d_{\alpha} r_{\alpha} = 0$,
for some scalars $d_{\alpha}$.
The idea is to put a partial order $\sqsupset$ on
$CA^{k,-}$
and show by induction on $\sqsupset$ that for all $\alpha \in CA^k$,
the coefficient $d_\alpha = 0$. 
 
We define $\sqsupset$ as follows. We take $\noop$ to be the minimal
element of $\sqsupset$. For actions $\alpha = \dop{\phi_A}; \gamma_\beta$ and 
$\alpha' = \dop(\phi_{A'}); \gamma_{\beta'}$ (where we take
$\gamma_\beta$ to be $\noop$ if 
$\alpha = \dop(\phi_A)$ and similarly for $\gamma_{\beta'}$), $\alpha \sqsupset \alpha'$ iff either
(1) $A \supsetneq A'$,
(2) $A = A'$,
$\beta \ne \noop$, and $\beta' = \noop$, or (3) $A = A'$,
$c_{\beta}(a) \sqsupset c_{\beta'}(a)$ for all atoms $a$.

We show that $d_\alpha = 0$ by induction assuming that $d_{\alpha'} = 0$
for all actions $\alpha' \in CA^{k,-}$ such that $\alpha \sqsupset \alpha'$.
For the base case, $\alpha = \noop$. 
Consider the $k$-progress function $g_\noop^k$ such that $g_\noop^k(t)
= t$ for all nodes $t$ in a $k$-olt. Note that 
$g_\noop(s^k,\id) = g^k_\noop$ for all $k$-olts $s^k$.
It is easy to see that if $\beta$ has the form $\dop(\phi_A)$ or
$\dop(\phi_A);\gamma_{\beta'}$, 
then for all $k$-olts $s^k$,  $g_{\beta,s^k} \ne g_\noop^k$ (since for
the root $r$ of $s^k$, $g_{\beta,s^k}(r) \ne r$).
  Thus, the entry of $r_\noop$ corresponding to the column
  $(s^k,g_\noop^k)$ is 1, while the entry of $d_\beta$ for $\beta \ne \noop$
  corresponding   to this column is 0.  
It follows that $d_\noop = 0$. 

For the general case, suppose that we have an arbitrary action $\alpha \ne
\noop$ in
$CA^{k,-}$  and $d_{\alpha'} = 0$ for all $\alpha' \in CA^k$ such that
$\alpha \sqsupset \alpha'$.
We now define a $k$-olt $s^{k,\alpha} \in T^k_a$ such that if
$g_{\alpha',s^{k,\alpha}} = g_{\alpha,s^{k,\alpha}}$ and $\alpha \ne \alpha'$, then $\alpha \sqsupset \alpha'$, so
$d_{\alpha'} = 0$ by the induction hypothesis.  Once we show this, it
follows that $d_{\alpha} = 0$ 
(since otherwise the entry in $\sum_{\alpha'} d_{\alpha'} r_{\alpha'}$ corresponding to 
$g_{\alpha,s^{k,\alpha}}$ would be nonzero). We construct $s^{k,\alpha}$ by induction
on the depth of $\alpha$.
If $\alpha$ has depth $1$ and is not $\noop$, it must
be of the form 
$\dop(\phi_A)$ for some set $A$ of atoms. Suppose that $b \in A$.  Let
the total order at the 
root of $s^{k,\alpha}$ be such that the final elements in the order are the
elements in $A$, and $b$ is 
the first of these.  For example, if $A = \{b,c,d\}$, we could consider
an order where the final three elements are $b$, $c$, and $d$ (or
$b$, $d$, and $c$).  Note that if $r$ is the root of $s^{k,\alpha}$, then
$g_{\alpha,s^{k,\alpha}}(r)$ is the child $t_b$ of $r$ labeled $b$.  
Now consider an action $\alpha'$ of the form
$\dop(\phi_{A'}); \beta$ ($\beta$ may be $\noop$). If $A'$ contains an element
not in $A$, then $g_{\alpha,s^{k,\alpha}}(r) \ne t_b$ (because there will be an
atom in $A'$ that is greater than $b$ in the total order at $r$).
If $A' \subset A$, then $\alpha \succ \alpha'$, as desired.  And if $A
= A'$ and 
$\alpha \ne \alpha'$, then $\alpha' = \phi_A;\gamma_\beta$ and $\beta \ne
\noop$, so it is easy to 
see that $g_{\alpha',s^{k,\alpha}}(r) \ne t_b = g_{\alpha,s^{k,\alpha}}(r)$.

Suppose that $m>1$ and we have constructed $s^{k,\beta}$ for all
actions $\beta \in CA^{k,-}$ of depth less than $m$.  We now show how to construct
$s^{k,\alpha}$ for actions $\alpha \in CA^{k,-}$ of depth $m$ that are not
of depth $m-1$ .  This means that $\alpha$ must have the form
$\dop(\phi_A);\beta$.  
We construct the total order at $r$ as above,
and at the subtree of $s^{k,\alpha}$ whose root is the child of
$r$ labeled $a$, we use the same orderings as in
$s^{k-1,c_{\beta}(a)}$, which by the induction hypothesis we have already
determined. It now follows easily from the induction hypothesis that
if $g_{\alpha',s^{k,\alpha}} = g_{\alpha,s^{k,\alpha'}}$ and $\alpha \ne \alpha'$, then $\alpha \sqsupset \alpha'$.
This completes the argument for (\ref{eq3}).

The argument above gives us a representation theorem for each $k$ that
works for actions of depth $k$.
However, we are interested in a single representation theorem that
works for all actions of all depths simultaneously. The first step is to make the state-dependent
utility functions $v^1, v^2, \ldots$ that we began with
(one utility function for each $k$ in the argument above)
\emph{$v$-compatible}, in the sense that if $\alpha$ is a depth-$k$
action and $k' > k$, then $v^k(a,c_{\alpha}(a)) =
v^{k'}(a,c_{\alpha}(a))$. That is, 
we want to construct a sequence $(v^1, v^2, v^3, \ldots)$
of $v$-compatible utility functions, each of which satisfies (\ref{eq1}).
We proceed as follows.

We can assume without loss of generality that each utility function has
range in $[0,1]$, by applying an affine transformation.  (Doing this
would not affect (\ref{eq1}).)  For each utility function $v^k$
let $v^{ki}$, for $i \le 
k$, be the restriction of $v^k$ to actions of depth $i$. Thus, $v^{kk} =
v^k$.  Now consider the sequence $v^{11}, v^{21}, v^{31}, \ldots$   It must
have a convergent subsequence, say $v^{m_1,1}, v^{m_2,1}, v^{m_3,1},
\ldots$. Say it converges to $w^1$. Now consider the subsequence
$v^{m_2,2}, v^{m_3,2}, \ldots$.  (We omit $v^{m_1,2}$, since we may
have $m_2=1$, in which case $v^{m_1,2}$ is not defined.) It too has
a convergent 
subsequence.  Say it converges to $w^2$.  Continuing this 
process, for each $k$, we can find a convergent subsequence, which is
a subsequence of the sequence we found for $k-1$.  It is easy to check
that the limits $w^1, w^2, w^3, \ldots$ of these convergent
subsequences satisfy (\ref{eq1}) and are $v$-compatible (since, in
general, $v^{ki}$ is $v$-compatible with $v^{kj}$ for $i,j \le k$).
For the remainder of this discussion, we assume without loss of
generality that the utility functions
in the sequence $v^1, v^2, \ldots$ are $v$-compatible.

Note that it follows easily from our definition that probability measures in
the sequence $\Pr^1, \Pr^2, \ldots$ are \emph{$\Pr$-compatible} in the following
sense:
If $k' > k$, $(s^k,\id) \in \Omega^k$, and
$E^{k'}(s^k,\id)$ consists of all the pairs $(t^{k'},\id)$ such that
$s^k$ is the projection of $t^{k'}$ onto depth $k$, then $\Pr^{k}(s^k,\id) =
\Pr^{k'}(E^{k'}(s^k,\id))$.
We will also want a third type of compatibility among the
utility functions. To make this precise, define a $k'$-progress function
$g$ to be \emph{$k$-bounded} for $k < k'$ if 
for all nodes $t$ of depth $\le k$, we have that $g(t)$ has depth $\le
k$, and if the depth of $t$ is greater than $k$, then $(t) = t$.
Note that if $\alpha$ is a depth-$k$ action, then $g_{\alpha,s^{k'}}$ is
$k$-bounded. If $k' > k$ and $g$ is a $k$-bounded $k'$-progress
function, then $g$ has an obvious projection to a $k$-progress function.
We want the utility functions in the sequence $u^1, u^2, \ldots$
that satisfies (\ref{eq2}) to be \emph{$u$-compatible} in the
following sense: if $g'$ is a $k'$-progress function that is
$k$-bounded, $g$ is the projection of $g'$ to a $k$-progress function,
and $s^k$ is the projection of $t^{k'}$ onto depth $k$, then
$u^k(s^k,g) = u^{k'}(t^{k'},g')$.
We can assume without loss of generality that the utility functions
in the sequence $u^1, u^2, 
\ldots, $ are $u$-compatible. For given a sequence $u^1, u^2, \ldots$,
define the sequence $w^1, w^2, \ldots$ as follows. Let $w^1 = u^1$.
Suppose that we have defined $w^1, \ldots, w^k$.  If the
$(k+1)$-progress function $g'$ is $k$-bounded, 
define $w^{k+1}(t^{k+1},g')
= w^k(s^k,g)$, where $s^k$ is the projection of $t^{k+1}$ onto depth $k$
and $g$ is the projection of $g'$ to a $k$-progress function;
if $g$ is not $k$-bounded, define $w^{k+1}(t^{k+1},g') =
u^{k+1}(t^{k+1},g')$.  Clearly the sequence $w^1, w^2, \ldots$ is
$u$-compatible. Moreover, it is easy to check that $(\Pr^k,w^k)$
satisfies (\ref{eq2}).  
  
We are now ready to define a single state space.
Define an
\emph{$\infty$-olt} just like a $k$-olt, except that now the tree is
unbounded, rather than having depth $k$.  
Let $\Omega^\infty$ consist of all pairs $(s^\infty,g)$, where
$s^\infty$ is an $\infty$-olt and $g$ is a $k$-bounded progress
function for some $k$. This will be our state space. Define
$E^{\infty}(s^k,\id)$ by obvious analogy to $E^{k'}(s^k,\id)$: it consists of all
the pairs $(t^{\infty},\id)$ such that $t^\infty$ extends $s^k$.
Then, by Carath{\' e}odory's extension theorem \cite{Ash70}
there
is a measure $\Pr^{\infty}$ on the smallest $\sigma$-algebra extending
the algebra generated by the sets $E^{\infty}(s^k,\id)$ which agrees
with $\Pr^k$ for all $k$ (i.e., $\Pr^k(s^k,\id) = 1/|T^k| =
\Pr^{\infty}(E^{\infty}(s^k,\id))$. Let
$u^{\infty}$ be defined by taking $u^\infty(s^{\infty},g) =
u^k(s^k,g^k)$ if $g$ is $k$-bounded and $s^k$ is the unique $k$-olt
that $s^\infty$ extends. It is easy to check that this is
well-defined (note that if $g$ is $k$-bounded then $g$ is $k'$-bounded
for $k' > k$, so there is something to check here). Finally, 
it is easy to check that for a depth-$k$ action $\alpha$, we have that
the expected utility of $\alpha$ is
$$\sum_{(s^k,g) \in \Omega^k} {\Pr}^\infty(E^\infty(s^k,\id))
u^\infty(s^k,g_{\alpha,s^k}) = \sum_a v^k(a,c_{\alpha}(a)),$$
giving us the desired result.
\end{proof}

\section{Conclusion and Future Work}


We have extended the results of \cite{BH21} to allow for actions that
are composed of sequences of steps, and proved a representation
theorem in this setting.
More precisely, we have shown that when an agent's language-based
preferences satisfy a suitably formulated cancellation axiom, they are
acting as if they are an expected utility maximizer with respect to
some background state space $\Omega$, a probability and utility over
$\Omega$, and a selection function on $\Omega$ that serves to
``disambiguate'' the results of actions described in the language. 
Allowing for (possibly unbounded) sequences
of steps made the proof significantly more complicated.

In \cite{BH21}, we also considered axioms regarding the preference
order $\succeq$ that restricted properties of
the selection function in ways that are standard in 
the literature on counterfactual conditions (e.g., being
\emph{centered}, so that $\sel(\omega,\phi) = \omega$ whenever $\omega
\models \phi$).  Although we have not checked details yet, we believe
it will be straightforward to provide axioms that similarly restrict
the selection function in our setting, and to extend the representation
theorem appropriately. 

We also believe it is of interest in some contexts to consider more
complex sequential actions, such as ``do $\phi$ until $\psi$''. This
opens the door for potentially \emph{non-terminating} actions, which
of course will add further complexity to the analysis.

Finally, and perhaps most urgently, while the cancellation axiom is
quite amazing in the power it has, it is not particularly
intuitive.
As shown in \cite{BEH06}, more intuitive axioms can be derived from
cancellation, such as transitivity of the relation $\succeq$ or the classic
principle of \emph{independence of irrelevant alternatives} (see
\cite{Savage}). In order to bring the technical results of this
project more in line with everyday intuitions about preference, it
would be very beneficial to ``factor'' the cancellation axiom into
weaker, but easier to intuit, components. This is the subject of
ongoing research.

\section*{Acknowledgments}
Joe Halpern was supported in part by 
ARO grant W911NF-22-1-0061, MURI grant W911NF-19-1-0217, AFOSR grant
FA9550-12-1-0040, and a
grant from the Algorand Centres of Excellence program managed by
Algorand Foundation.  Any opinions, findings, and conclusions or
recommendations expressed in this material are those of the author(s)
and do not necessarily reflect the views of the funders.

\bibliographystyle{eptcs}
\bibliography{z,joe}
\end{document}